\def\be{\begin{equation}}
	\def\ee{\end{equation}}
\def\ba{\begin{array}}
	\def\ea{\end{array}}

\documentclass[prl,showpacs,twocolumn,amsmath]{revtex4}
\usepackage{amsmath}
\usepackage{amsfonts}
\usepackage{mathrsfs}
\usepackage{amssymb}
\usepackage{pifont}
\usepackage{epsfig,subfigure,dsfont,amsthm,amsbsy,mathrsfs,amscd}
\usepackage{epstopdf}
\usepackage{color}
\def\qed{\leavevmode\unskip\penalty9999 \hbox{}\nobreak\hfill
	\quad\hbox{\leavevmode  \hbox to.77778em{%
			\hfil\vrule   \vbox to.675em%
			{\hrule width.6em\vfil\hrule}\vrule\hfil}}
	\par\vskip3pt}
\usepackage{leftidx}
\newtheorem{theorem}{Theorem}

\input amssym.def
\begin{document}
	\title{\large\bf Note on quantum discord}

	\author{Yiding Wang, Xiaofen Huang and Tinggui Zhang$^{\dag}$}
	\affiliation{ School of Mathematics and Statistics, Hainan Normal University, Haikou, 571158, China \\
		$^{\dag}$ Correspondence to tinggui333@163.com}
	
	\bigskip
	\bigskip
	
\begin{abstract}
Quantum discord goes beyond entanglement and exists in a wide range of quantum states that may be separable, playing a crucial role in quantum information tasks. In this paper, we firstly proposed a zero-discord criterion for two-qubit system based on the partial transposition of density matrix, and then extended it to the qubit-qudit system. By detailed
examples we demonstrate the effectiveness of these criteria in detecting discord. Moreover, we provide an analytical lower bound of geometric quantum discord(GQD) using eigenvalue vectors of density matrix. Finally, we presented a one-way work deficit lower bound based on our lower bound of GQD.

\end{abstract}
	
 \maketitle
	
\section{I. Introduction}
One of the most prominent features of quantum mechanics is the existence of quantum correlations between quantum systems. The quantum entanglement \cite{rpmk} is widely studied and considered as a key resource in many quantum information processing such as quantum communications \cite{cgcraw,rh}, quantum computing \cite{asc,ar}, quantum cryptography \cite{a,ngwh} and quantum simulation \cite{s}. Besides quantum entanglement, the quantum discord \cite{oz,hv} has been discovered as a more fundamental quantum correlation than entanglement.

As a novel quantum correlation, discord was first proposed by Ollivier and Zurek \cite{oz,z}, and analyzed in detail by Vedral et al. \cite{hv,dvb}. Discord also illustrates its quantum advantages in diversified tasks such as mixed-state quantum computing \cite{kl,dsc}, quantum state merging \cite{md,cabm}, remote state preparation \cite{dlm} and optimal assisted state discrimination \cite{libo}. 

Let $\mathcal{H}_A$ ($\mathcal{H}_B$) be the Hilbert space associated with systems $A$ ($B$). For a bipartite quantum state $\rho\in\mathcal{H}_A\otimes\mathcal{H}_B$, the quantum mutual information is given by
$$
\mathcal{I}(\rho)=\mathcal{S}(\rho_A)+\mathcal{S}(\rho_B)-\mathcal{S}(\rho),
$$
where $\rho_A$ $(\rho_B)$ is the reduced density matrix in $\mathcal{H}_A$ $(\mathcal{H}_B)$, and $\mathcal{S}(\rho)=-Tr(\rho\log_2\rho)$ stands for the von Neumann entropy. Under a von Neumann measurement on subsystem $A$ given by projective operators $\{E_i^A\}$, the conditional state $\rho_i$ associated with the measurement outcome $i$ is given by
$$
\rho_i=\frac{1}{p_i}(E_i^A\otimes I)\rho(E_i^A\otimes I),
$$
where the probability $p_i=Tr[(E_i^A\otimes I)\rho(E_i^A\otimes I)]$ with $I$ the identity operator. The quantum mutual information related to this measurement is defined as \cite{l},
$$
\mathcal{I}(\rho|\{E_i^A\})=\mathcal{S}(\rho_B)-\mathcal{S}(\rho|\{E_i^A\}),
$$
where $\mathcal{S}(\rho|\{E_i^A\})=\sum_{i}p_i\mathcal{S}(\rho_i)$ is the quantum conditional entropy. The corresponding classical correlation is given by \cite{oz,hv,l},
$$
\mathcal{C}_A(\rho)=\mathop{\sup}_{\{E_i^A\}}\mathcal{I}(\rho|\{E_i^A\}),
$$
and the quantum discord is defined to be
$$
D_A(\rho)=\mathcal{I}(\rho)-\mathcal{C}_A(\rho).
$$

The quantum discord has been extensively investigated with potential applications in many physical systems \cite{bc,wdj,bzsm,hf,xg,ll,bff,ara,ylzh,gb,yzco,zyco,njby,jzy,rlb,bcxb,jxnj,szy,bqy,mbcp,lf}. Let us briefly introduce these results. According to the relationship between quantum entanglement and discord, the latter can be studied using the methods of the former. Like entanglement witnesses, one way to detect the discord is to use discord witnesses \cite{bc,ara,ylzh,gb,yzco,zyco,njby}. In \cite{bc}, Bylicka et al. proposed a witness for a quantum discord: if a $2\times N$ state is not strong positive partial transpose\,(SPPT) it must contain nonclassical correlations measured by quantum discord. The authors in \cite{ara} studied the quantum discord of two qubit X-states and investigated the relationship among the discord, classical correlation and entanglement. In \cite{gb} the authors proposed a method for the local detection of quantum correlations to detect the discord in bipartite systems when access is restricted to only one of the subsystems. A single observable to witness the non-zero quantum discord of an unknown state with four copies was presented in \cite{zyco}, wherein the expectation value of this observable provides a necessary condition of non-zero quantum discord for higher finite-dimensional bipartite systems. More recently, for any $N$-partite qudit states the authors in \cite{jzy} show that there exists such a hierarchy: genuine multipartite total correlations (GMT) $\supseteq$ genuine multipartite discord (GMD) $\supseteq$ genuine multipartite entanglement (GME) $\supseteq$ genuine multipartite steering (GMS) $\supseteq$ genuine multipartite nonlocality (GMNL). A generalization of quantum discord to multipartite systems is proposed in \cite{rlb}, there are some generalizations of it \cite{bcxb,jxnj}. The authors in \cite{szy} show that the partial transpose of density matrix can be used to detect discord, by surveying the changes of the spectrum of the density matrix after partial transpose. In \cite{bqy} the authors find that the quantum discord with weak measurements decays in a monotonic fashion and show that discord might exhibit a sudden change only for the projective measurements. As for the quantification of discord \cite{mbcp}, the geometric quantum discord (GQD), the minimum Hilbert-Schmidt distance between a given state and the set of discord-free states, is a widely used measure of discord \cite{dvb,lf}. Nevertheless, the analytical calculation of discord is challenging for general quantum states.

The paper is organized as follows. In the second section, we provide a discord-free criterion for two-qubit systems based on the partial transpose of the density matrix. In the third section, we extend this criterion to qubit-qudit systems. In the fourth section, we derive an analytical lower bound of geometric quantum discord by using eigenvalue vectors. In the fifth section, we propose a lower bound of one-way work deficit based on our lower bound of GQD. We summarize and discuss our conclusions in the last section.

\section{II. Partial transpose based discord criterion for two-qubit systems}
Partial transpose was originally used to detect quantum entanglement \cite{p}. It has been demonstrated that all strong PPT states are separable \cite{cjk}. In this section we present
a partial transpose based criterion to detect the discord of two-qubit states.
Any zero-discord state $\rho$, $D_A(\rho)=0$, can be written as \cite{mbcp},
\begin{equation}\label{e3}
		\rho_{AB}=\sum_{k}p_k|k\rangle\langle k|\otimes\rho_k^{\tau_B},
\end{equation}

where $\{|k\rangle\}$ is an orthogonal basis in $H_A$.
Let $\rho^{\tau_B}$ be the partial transposed matrix of $\rho$ with respect to subsystem B, where $(\rho^{\tau_B})_{ij,kl}=(\rho)_{il,kj}$. We have the following conclusion.

\begin{theorem}
If $D_A(\rho)=0$, then $PM_{nm}(\rho)=PM_{nm}(\rho^{\tau_B})$ for all $1\leq n<m\leq4$, where $PM_{nm}(\rho)$ represents the second-order principal minor composed by the elements that intersect the $n$-th row, $m$-th row and the $n$-th column, $m$-th column of the density matrix $\rho$.
\end{theorem}

\begin{proof}
For a general density matrix,
$$
\rho=\left(\begin{array}{cccc}
	\rho_{11} & \rho_{12} & \rho_{13} & \rho_{14} \\
	\rho_{12^*} & \rho_{22} & \rho_{23} & \rho_{24} \\
	\rho_{13^*} & \rho_{23^*} & \rho_{33} & \rho_{34} \\
	\rho_{14^*} & \rho_{24^*} & \rho_{34^*} & \rho_{44}
\end{array}\right),
$$
we have 
$$
\rho^{\tau_B}=\left(\begin{array}{cccc}
	\rho_{11} & \rho_{12^*} & \rho_{13} & \rho_{23} \\
	\rho_{12} & \rho_{22} & \rho_{14} & \rho_{24} \\
	\rho_{13^*} & \rho_{14^*} & \rho_{33} & \rho_{34^*} \\
	\rho_{23^*} & \rho_{24^*} & \rho_{34} & \rho_{44}
\end{array}\right).
$$
Therefore, from the definition of the principal minor, the invariance of $PM_{12}$, $PM_{13}$, $PM_{24}$ and $PM_{34}$ is trivial. We obtain $PM_{14}(\rho)=\rho_{11}\rho_{44}-|\rho_{14}|^2$ and $PM_{14}(\rho^{\tau_B})=\rho_{11}\rho_{44}-|\rho_{23}|^2$. From (\ref{e3}) $\rho$ can be written as
$$
\rho=\left(\begin{array}{cc}
	a_{11} & a_{12} \\
	a_{13} & a_{14} \\
\end{array}\right)\otimes
\left(\begin{array}{cc}
	b_{11} & b_{12} \\
	b_{12^*} & b_{14} \\
\end{array}\right)+
\left(\begin{array}{cc}
	a_{21} & a_{22} \\
	a_{23} & a_{24} \\
\end{array}\right)\otimes
\left(\begin{array}{cc}
	b_{21} & b_{22} \\
	b_{22^*} & b_{24} \\
\end{array}\right),
$$
where due to the property of orthonormal basis, $\sum_{k=1}^{2}|k\rangle\langle k|=I$, we have $a_{12}+a_{22}=0$. Thus
$|\rho_{14}|=|a_{12}(b_{12}-b_{22})|$ and $|\rho_{23}|=|a_{12}(b_{12}^*-b_{22}^*)|$.
According to the fact that the modulus of the product of complex numbers is equal to the product of the modulus, and that the modules of two complex numbers that are conjugate to each other are equal, we have $|\rho_{14}|=|\rho_{23}|$, namely, $PM_{23}(\rho)=PM_{23}(\rho^{\tau_B})$.
\end{proof}

Thus, from Theorem 1, one can know that after the operation of partial transposition, the change in second-order principal minors implies a non-trivial discord. We illustrate the usefulness and conciseness of our criterion through the following example.

Let us consider the Werner state \cite{w}
\begin{equation}
\rho_a=a|\psi^-\rangle\langle\psi^-|+\frac{1-a}{4}I_4,\nonumber
\end{equation}
where $|\psi^-\rangle=\frac{1}{\sqrt{2}}(|01\rangle-|10\rangle)$ is a maximally entangled state and $0\leq a\leq1$.

It is clear that as long as $a\neq0$, $\rho_a$ violates Theorem 1, which means that for any $0<a\leq1$, $\rho_a$ has quantum discord. According to the PPT criterion or realignment criterion, when $a>\frac{1}{3}$, $\rho_a$ is entangled. This indicates that entangled states must contain discord, but it is not necessary for the converse. The results obtained from our discord criterion are consistent with the results in \cite{ara}, where when $a>0$, $\rho_a$ has a non-trivial discord, and when $a>\frac{1}{3}$, the concurrence $C(\rho_a)>0$.

The Bell states given as $|\psi^{\pm}\rangle=\frac{1}{\sqrt{2}}(|01\rangle\pm|10\rangle)$ and $|\phi^{\pm}\rangle=\frac{1}{\sqrt{2}}(|00\rangle\pm|11\rangle)$. For these particular states, discord and any entanglement measure coincide and are equal to the maximum value of the correlation. The authors proposed in \cite{ara} that if we mix any two Bell states, the entanglement of formation \cite{bdsw,wkw} and quantum discord are equal. For example, consider $\rho_b=b|\psi^-\rangle\langle\psi^-|+(1-b)|\phi^+\rangle\langle\phi^+|$, by calculation, Theorem 1 will not be violated only when $b=\frac{1}{2}$. That is to say, when $b$ takes other values, the quantum state $\rho_b$ has non-trivial discord. This result do agree with the above analysis. It is noteworthy that convex combination $\frac{1}{2}(\rho_1+\rho_2)$ of arbitrary two Bell states has vanishing discord, this is a magical result. See Fig. 1.
\begin{figure}[htbp]
	\centering
	\includegraphics[width=0.5\textwidth]{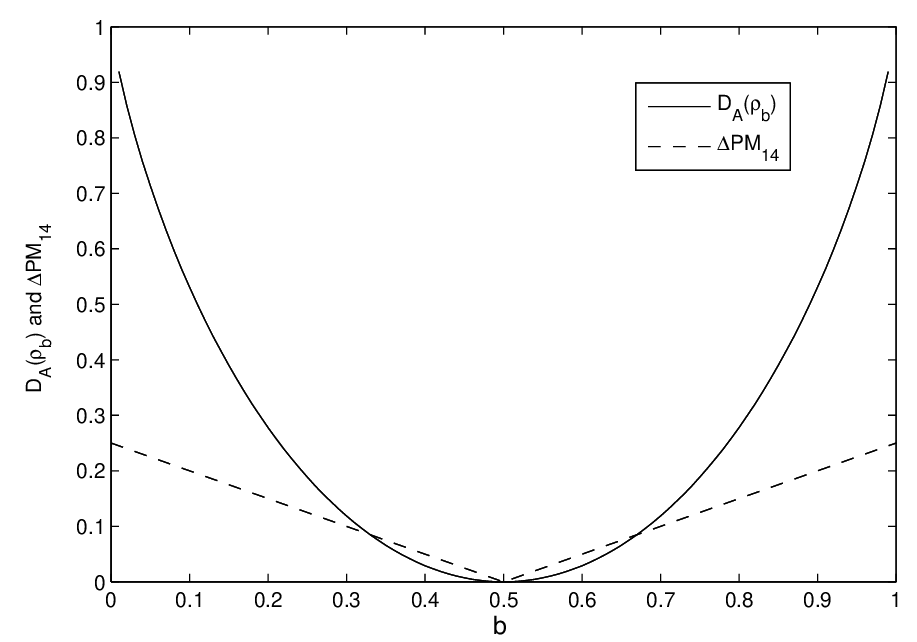}
	\vspace{-1.2em} \caption{As shown in the figure, the solid line represents $D_A(\rho_b)$, and the dashed line represents the change in the principal minor $\Delta PM_{14}=|PM_{14}(\rho)-PM_{14}(\rho_b^{\tau_B})|$. When $b=\frac{1}{2}$, $\rho_b$ is a zero-discord state.} \label{Fig.1}
\end{figure}
\section{III. Discord criterion in qubit-qudit system}
We are now considering extending our criterion to $2\otimes n$ system. For simplicity, one consider the density matrix $\rho$ in $2\otimes3$ system, generally $\rho$ can also be written as
$$
\rho=\left(\begin{array}{cccccc}
	\rho_{11} & \rho_{12} & \rho_{13} & \rho_{14} & \rho_{15} & \rho_{16} \\
	\rho_{12^*} & \rho_{22} & \rho_{23} & \rho_{24} & \rho_{25} & \rho_{26} \\
	\rho_{13^*} & \rho_{23^*} & \rho_{33} & \rho_{34} & \rho_{35} & \rho_{36} \\
	\rho_{14^*} & \rho_{24^*} & \rho_{34^*} & \rho_{44} & \rho_{45} & \rho_{46} \\
	\rho_{15^*} & \rho_{25^*} & \rho_{35^*} & \rho_{45^*} & \rho_{55} & \rho_{56} \\
	\rho_{16^*} & \rho_{26^*} & \rho_{36^*} & \rho_{46^*} & \rho_{56^*} & \rho_{66} \\
\end{array}\right).
$$
After deleting the first and fourth rows as well as the first and fourth columns of matrix $\rho$, we have
$$
\rho_{/14}=\left(\begin{array}{cccc}
	\rho_{22} & \rho_{23} & \rho_{25} & \rho_{26} \\
	\rho_{23^*} & \rho_{33} & \rho_{35} & \rho_{36} \\
	\rho_{25^*} & \rho_{35^*} & \rho_{55} & \rho_{56} \\
	\rho_{26^*} & \rho_{36^*} & \rho_{56^*} & \rho_{66} \\
	
\end{array}\right).
$$
On the one hand, it is clear that $\rho_{/14}$ is still a Hermitian matrix (generally, it is no longer a quantum state) and the set of all its second-order principal minors is a subset of all second-order principal minors of matrix $\rho$.

If we delete the second and fifth rows as well as the second and fifth columns of matrix $\rho$, we obtain
$$
\rho_{/25}=\left(\begin{array}{cccc}
	\rho_{11} & \rho_{13} & \rho_{14} & \rho_{16} \\
	\rho_{13^*} & \rho_{33} & \rho_{34} & \rho_{36} \\
	\rho_{14^*} & \rho_{34^*} & \rho_{44} & \rho_{46} \\
	\rho_{16^*} & \rho_{36^*} & \rho_{46^*} & \rho_{66} \\
	
\end{array}\right).
$$
By deleting the third, sixth row, and the third, sixth column, one can get
$$
\rho_{/36}=\left(\begin{array}{cccc}
	\rho_{11} & \rho_{12} & \rho_{14} & \rho_{15} \\
	\rho_{12^*} & \rho_{22} & \rho_{24} & \rho_{25} \\
	\rho_{14^*} & \rho_{24^*} & \rho_{44} & \rho_{45} \\
	\rho_{15^*} & \rho_{25^*} & \rho_{45^*} & \rho_{55} \\
	
\end{array}\right).
$$
Making a discussion similar to $\rho_{/14}$, we can obtain that all second-order principal minors of $\rho_{/14}$, $\rho_{/25}$, and $\rho_{/36}$ traverse all second-order principal minors of $\rho$.

Based on the above discussion, for the $2\otimes3$ dimensional system, we provide the following results.
\begin{theorem}
For a $2\otimes3$ quantum state $\rho$, if $D_A(\rho)=0$, all its second-order principal minors remain unchanged after partial transposition.
\end{theorem}
\begin{proof}
$D_A(\rho)=0$ means that there exists a set of orthonormal bases $\{|l\rangle\}$ in $H_A$ such that
\begin{equation}\label{e8}
	\rho=\sum_{l=1}^2q_l|l\rangle\langle l|\otimes\rho_l^B,
\end{equation}
where $\rho_l^B$ are 3-order density matrices in $H_B$. Based on the above derivation process, we can obtain
\begin{eqnarray}
\rho_{/14}&=&\sum_{l}q_l|l\rangle\langle l|\otimes\rho_{l/1}^B,\\
\rho_{/25}&=&\sum_{l}q_l|l\rangle\langle l|\otimes\rho_{l/2}^B,\\
\rho_{/36}&=&\sum_{l}q_l|l\rangle\langle l|\otimes\rho_{l/3}^B.
\end{eqnarray}
Where $\rho_{l/i}^B$ is the second-order Hermitian matrix obtained by erasing the i-th row and i-th column of $\rho_l^B$. Make a similar deduction to the proof process of Theorem 1 for $\rho_{/14}$, $\rho_{/25}$, and $\rho_{/36}$, it can be concluded that all their second-order principal minors are invariant under partial transposition. Note that all second-order principal minors of $\rho_{/14}$, $\rho_{/25}$, and $\rho_{/36}$ traverse all second-order principal minors of $\rho$.
Therefore, the theorem is proven.
\end{proof}
The results of our derivation in $2\otimes 3$ quantum system can be naturally extended to $2\otimes n$ system. To perform similar operations, we need to eliminate the $2n-4$ rows and $2n-4$ columns of the $2\otimes n$ density matrix. In fact, this is equivalent to selecting the fourth-order principal sub-matrix of the density matrix according to certain rules.

\begin{theorem}
For a quantum state $\rho$ in $2\otimes n$ system, if $D_A(\rho)=0$, then all second-order principal minors of matrix $\rho$ remain unchanged under partial transposition.
\end{theorem}
\begin{proof}
Without losing generality, we assume that $1\leq i<j<k<l\leq 2n$, and the principal submatrix composed of the $i$-th, $j$-th, $k$-th, $l$-th rows, as well as the $i$-th, $j$-th, $k$-th, $l$-th columns of matrix $\rho$ is represented as $\rho_s^{i,j,k,l}$. Our rule is that for any $i$ and $j$ that satisfy $1\leq i<j\leq n$, the corresponding fourth-order principal submatrix is selected as $\rho_s^{i,j,i+n,j+n}$. For any given $i$ and $j$, after partial transposition, all second-order principal minors of $\rho_s^{i,j,i+n,j+n}$ remain unchanged. It is easy to see that the second-order principal minors of all the principal submatrices selected according to this rule traverse all the second-order principal minors of matrix $\rho$. Thus, the proof of the theorem is completed.
\end{proof}
Consider the quantum states with two real parameters $u$ and $c$ in $2\otimes3$ system given in \cite{dl}
\begin{align}\label{ex1}
\rho_{u,c}=&u(|02\rangle\langle02|+|12\rangle\langle12|)+f(|\phi^+\rangle\langle\phi^+|+|\phi^-\rangle\langle\phi^-|\nonumber\\
          &+|\psi^+\rangle\langle\psi^+|)+c|\psi^-\rangle\langle\psi^-|,
\end{align}
where parameter $f$ is dependent on $u$ and $c$ by the unit trace condition,
\begin{equation*}
2u+3f+c=1.
\end{equation*}
According to Eq.(\ref{ex1}), one can obtain that the range of parameters are $0\leq u\leq\frac{1}{2}$, $0\leq c\leq1$ and $0\leq f\leq\frac{1}{3}$. After partial transpose,
$$
\rho_{u,c}^{\tau_B}=\left(\begin{array}{cccccc}
	f & 0 & 0 & 0 & \frac{f-c}{2} & 0 \\
	0 & \frac{f+c}{2} & 0 & 0 & 0 & 0 \\
	0 & 0 & u & 0 & 0 & 0 \\
	0 & 0 & 0 & \frac{f+c}{2} & 0 & 0 \\
	\frac{f-c}{2} & 0 & 0 & 0 & f & 0 \\
	0 & 0 & 0 & 0 & 0 & u \\
\end{array}\right).
$$
All eigenvalues of $\rho_{u,c}^{\tau_B}$ are non negative if and only if $\frac{3f-c}{2}=\frac{1}{2}-u-c\geq0$. Thus, the range of the parameters $u$ and $c$ consists of two triangular regions: the PPT separable region satisfying $0\leq u+c\leq\frac{1}{2}$ and the NPT entangled region satisfying $\frac{1}{2}\leq u+c\leq 1$, see Fig. 2. By using our discord criterion, we can know that the quantum state $\rho_{u,c}$ is discord-free if and only if $c=f$. Obviously, this situation is included within the range of the PPT. This also means that discord exists in a wider range of quantum states that may be separable.
\begin{figure}[htbp]
	\centering
	\includegraphics[width=0.54\textwidth]{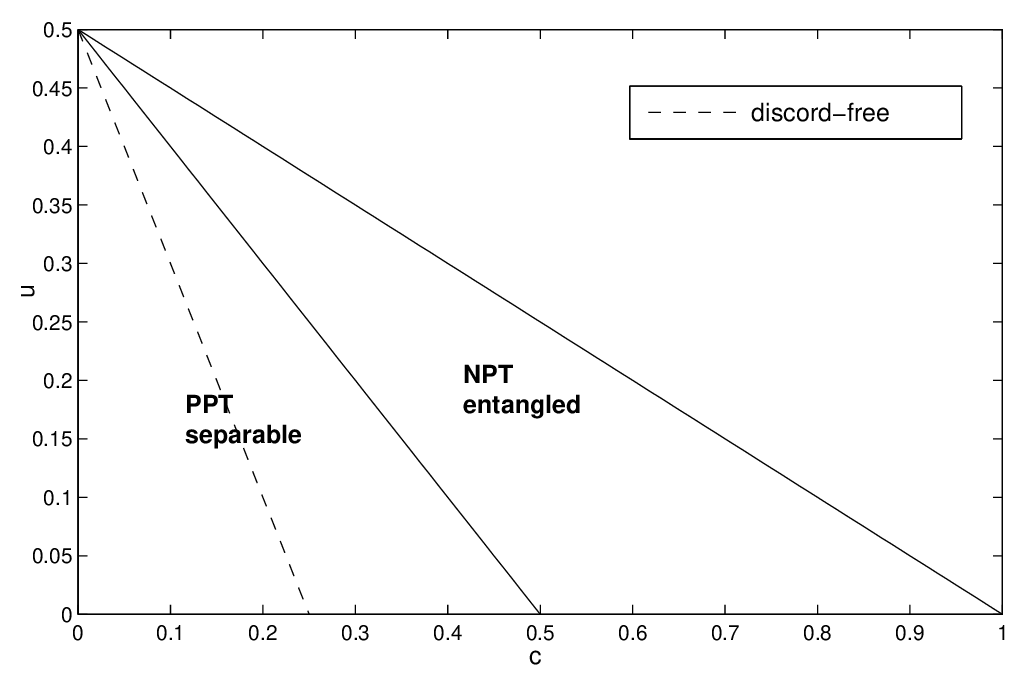}
	\vspace{-2em} \caption{As shown in the figure, the dashed line represents the distribution of all discord-free states for $\rho_{u,c}$, which are included in the PPT separable range.} \label{Fig.2}
\end{figure}
\section{IV. Analytic Lower Bound of GQD}
Geometric quantum discord (GQD) is a noteworthy discord measure defined as the minimum Hilbert-Schmidt distance between a given state $\rho$ and the set of zero-discord states represented by $\mathcal{D}_z$:
\begin{equation*}
D_{G}=\mathop{\min}_{\sigma\in\mathcal{D}_z}\|\rho-\sigma\|_2^2.
\end{equation*}
We arrange the eigenvalues of the $M\otimes N$ density matrix $\rho$ in descending order to form an eigenvalue vector $\overrightarrow{\lambda(\rho)}=(\lambda_1,\lambda_2,...,\lambda_{MN})$, which is abbreviated as $\overrightarrow{\lambda}$ without causing ambiguity. The authors in \cite{rb} point out that the following inequality holds:
\begin{equation}\label{e9}
\|\rho-\sigma\|_2^2\geq\|\overrightarrow{\lambda(\rho)}-\overrightarrow{\lambda(\sigma)}\|_2^2.
\end{equation}

Now, let us discuss the lower bound of GQD. Due to the invariance of the Hilbert-Schmidt norm under partial transposition, we have
\begin{eqnarray}
D_G&=&\frac{1}{2}\mathop{\min}_{\sigma\in\mathcal{D}_z}(\|\rho-\sigma\|_2^2+\|\rho-\sigma\|_2^2)\nonumber\\
   &=&\frac{1}{2}(\|\rho-\sigma_{min}\|_2^2+\|\rho-\sigma_{min}\|_2^2)\nonumber\\
   &=&\frac{1}{2}(\|\rho-\sigma_{min}\|_2^2+\|\rho^{\tau_B}-\sigma_{min}^{\tau_B}\|_2^2)\nonumber\\
   &\geq&\frac{1}{2}(\|\overrightarrow{\lambda}-\overrightarrow{\lambda}_{min}\|_2^2+\|\overrightarrow{\lambda}^{'}-\overrightarrow{\lambda}_{min}\|_2^2)\nonumber\\
   &=&\frac{1}{2}\sum_{i}[2\lambda_{min,i}^2-2(\lambda_i+\lambda_i^{'})\lambda_{min,i}+\lambda_i^2+\lambda_i^{'2}]\nonumber\\
   &:=&\frac{1}{2}\sum_{i}L_i,\nonumber
\end{eqnarray}
where $\sigma_{min}$ specifies the closet discord-free state of $\rho$ with respect to the 2-norm, the inequality is due to (\ref{e9}) and the fact that the spectrum of zero-discord state remains unchanged under partial transposition \cite{szy}, and $\overrightarrow{\lambda}^{'}$ represents the eigenvalues vector of $\rho^{\tau_B}$. It is clear that if $\rho$ is a zero-discord state, then $L_i=0$, $i=1,2,...,MN$. We first consider calculating the minimum value of $L_1$ within the range of $\lambda_{min,1}\in[\frac{1}{MN},1]$, with the minimum point still denoted as $\lambda_{min,1}$ and the minimum value denoted as $\mathcal{L}_1$. Then, calculate the minimum value of $L_2$ on $\lambda_{min,2}\in[\frac{1-\lambda_{min,1}}{MN-1},\min\{1-\lambda_{min,1},\lambda_{min,1}\}]$, with the minimum point still denoted as $\lambda_{min,2}$ and the minimum value denoted as $\mathcal{L}_2$. Repeat the above operation. If $j$ $(1\leq j\leq MN)$ causes $1-\sum_{i=1}^{j}\lambda_{min,i}=0$, then $\lambda_{min,j+1}=\lambda_{min,j+2}=...=\lambda_{min,MN}=0$. Each of the above steps only involves the problem of finding the minimum value of a quadratic function within a given interval, so it is always possible to do so. The main results of this section are presented below.
\begin{theorem}
For any $M\otimes N$ quantum state $\rho$, the lower bound of GQD is
\begin{equation*}
D_G\geq\frac{1}{2}\sum_{i=1}^{MN}\mathcal{L}_i.
\end{equation*}
\end{theorem}
We consider the maximum entangled state $|\phi^+\rangle=\sum_{i=1}^{d}\frac{1}{\sqrt{d}}|ii\rangle$, and the exact value of GQD is $1-\frac{1}{d}$. When $d=2$, we have $\overrightarrow{\lambda}^{'}=(\frac{1}{2},\frac{1}{2},\frac{1}{2},-\frac{1}{2})$ and $\overrightarrow{\lambda}=(1,0,0,0)$. By calculation, we can obtain
\begin{equation*}
D_G\geq \frac{1}{2}\sum_{i=1}^{4}\mathcal{L}_i=\frac{1}{2}(0.125+0.125+0.25+0.25)=0.375.
\end{equation*}
Our lower bound of 0.375 is larger than the lower bound given in \cite{szy}, where $L_{PPT}^D=L_{SIPT}^D=\frac{1}{3}$. When $d=3$, $\frac{1}{2}\sum_{i=1}^{9}\mathcal{L}_i=\frac{1}{2}=L_{PPT}^D>L_{SIPT}^D=\frac{11}{24}$.
\section{V. Lower Bound of one-way work deficit}
For an $M\otimes N$ quantum state $\rho_{AB}$, it was shown that the works extractable by closed operations (CO) and closed local operations and classical communication (CLOCC) are respectively $I_{CO}$ and $I_{CLOCC}$ \cite{ohhh,jau,aub,d}, given by
\begin{eqnarray*}
I_{CO}(\rho_{AB})&=&\log_2MN-S(\rho_{AB}),\\
I_{CLOCC}(\rho_{AB})&=&\log_2MN-\mathop{\min}_{\{\Pi_i^A\}}S(\rho_{AB}^{'}),
\end{eqnarray*}
where $\rho_{AB}^{'}=\sum_{i}\Pi_i^A\otimes\mathbb{I}_N^B\rho_{AB}\Pi_i^A\otimes\mathbb{I}_N^B$ and $\{\Pi_i^A\}$ represents the dephasing operation acting on Alice's side. Here, the minimization is performed over all projective measurements on the system at A. The difference between $I_{CO}$ and $I_{CLOCC}$ is defined as the one-way work deficit \cite{ohhh,bdsr}, given by
\begin{eqnarray*}
\mathcal{WD}_A&:=&I_{CO}-I_{CLOCC}\\
    &=&\mathop{\min}_{\{\Pi_i^A\}}S(\rho_{AB}\|\rho_{AB}^{'}).
\end{eqnarray*}
Where $S(\rho\|\sigma):=Tr(\rho\ln\rho-\rho\ln\sigma)$ is the relative entropy between the two quantum states $\rho$ and $\sigma$. It is clear that $\rho_{AB}^{'}$ is a zero-discord state for any given $\{\Pi_i^A\}$, below we derive the lower bound of $\mathcal{WD}_A$.
\begin{eqnarray*}
\mathcal{WD}_A&=&\mathop{\min}_{\{\Pi_i^A\}}S(\rho_{AB}\|\rho_{AB}^{'})\geq\mathop{\min}_{\sigma\in\mathcal{D}_z}S(\rho_{AB}\|\sigma)\\
              &=&S(\rho_{AB}\|\sigma_{min,re})\geq\frac{\|\rho-\sigma_{min,re}\|_1^2}{2\ln2}\\
              &\geq&\mathop{\min}_{\sigma\in\mathcal{D}_z}\frac{\|\rho-\sigma\|_1^2}{2\ln2}=\frac{\|\rho-\sigma_{min,1}\|_1^2}{2\ln2}\\
              &\geq&\frac{\|\rho-\sigma_{min,1}\|_2^2}{2\ln2}\geq\frac{1}{2\ln2}D_G\\
              &\geq&\frac{1}{4\ln2}\sum_{i=1}^{MN}\mathcal{L}_i.
\end{eqnarray*}
Where $\sigma_{min,re}$ and $\sigma_{min,1}$ specify the closest zero-discord states of $\rho$ in the sense of relative entropy and 1-norm, respectively. The second inequality is due to the quantum Pinsker inequality $S(\rho\|\sigma)\geq\frac{\|\rho-\sigma\|_1^2}{2\ln2}$ \cite{oh} and the fourth one is due to $\|\rho-\sigma\|_1\geq\|\rho-\sigma\|_2$ \cite{ccc}.

\section{VI. Conclusions and discussions}

Discord goes beyond entanglement and captures the nonclassical correlations that can exist between parts of a separable quantum state. Based on the partial transposition of the density matrix we have proposed a zero-discord criterion for two-qubit quantum states. The main advantage of this criterion is simplicity, as we only need to check whether the principal minors remain unchanged after partial transposition, and in specific examples, compared with other results \cite{ara}, it proves the accuracy of our criterion in capturing quantum discord. Then we can also get the necessary condition that the discord is zero of $2\otimes n$ system, where we used examples to analyse the relationship between entangled, separable, and discord-free, and once again confirmed a fact: quantum discord exists in a wider range of states that may be separable. Compared with Ref.\cite{bc}, it only needs to detect the second-order principal minor, and does not need to find the canonical factorization of the original density matrix.  Moreover, we have presented an analytical lower bound of geometric quantum discord for arbitrary bipartite quantum states, which give the way to determine quantitatively the degree of quantum discord without the specific form of the closet discord-free state of $\rho$ with respect to the Hilbert-Schmidt distance. By calculation, we have demonstrated that our lower bound is not only an effective approximation of the geometric quantum discord of the given states, but also a productive improvement on the known lower bound \cite{szy}. Finally, we have provided a lower bound of one-way work deficit based on our analytical lower bound of GQD. As a significant quantum correlation, discord plays a pivotal role in various quantum technologies. Our results can deepen our understanding of the methodologies for discord detection and quantization. 

\bigskip
{\bf Acknowledgments:} We are grateful for discussions with Prof. Shao-Ming Fei. This work is supported by the Hainan Provincial Graduate Innovation Research Program under Grant No. Qhys2023-386; the
specific research fund of the Innovation Platform for Academicians of Hainan Province under Grant No. YSPTZX202215; the Hainan Provincial Natural Science Foundation of China under Grant No.121RC539.

\smallskip
{\bf Data Availability Statement:} This manuscript has no associated data.

\end{document}